\documentclass[11pt]{article}

 \usepackage{varioref}
 \usepackage{wrapfig}%
\usepackage{setspace}
\usepackage[]{graphicx} 
\usepackage{amsfonts}
\usepackage{amsthm}
\usepackage{amsmath}    
\usepackage{amssymb}    
\usepackage{subfigure}
\usepackage{hyperref}
\usepackage{url}
\usepackage{psfrag}

\title{\bf Projection Operator in Adaptive Systems}
\author{Eugene Lavretsky\thanks{The Boeing Company, Huntington Beach CA, USA} \
Travis E. Gibson\thanks{Department of Mechanical Engineering, Massachusetts Institute of Technology,Cambridge MA 02139 USA}
\ and Anuradha M. Annaswamy\thanks{Department of Mechanical Engineering, Massachusetts Institute of Technology,Cambridge MA 02139 USA}}

\theoremstyle{plain}
\newtheorem{thm}{Theorem}
\newtheorem{lem}[thm]{Lemma}

\theoremstyle{definition}
\newtheorem{defn}[thm]{Definition}
\newtheorem{exmp}[thm]{Example}
\theoremstyle{remark}
\newtheorem*{rem}{Remark}

\newcommand{\proj}[2]{\text{Proj}({#1},{#2})}
\newcommand{\Proj}[3]{\text{Proj}({#1},{#2},{#3})}
\providecommand{\norm}[1]{\lVert#1\rVert}
\newcommand{\re}[0]{\mathbb R}

\begin{document}
\maketitle
\begin{abstract}
The projection algorithm is frequently used in adaptive control and this note presents a detailed analysis of its properties.
\end{abstract}

\section{Introduction}
These notes started in \cite{lav06_proj} as a personal communication from Eugene to colleagues in the field of adaptive control and summarized results from \cite{slo86,pom92,ioabook,kha00}. Properties of the projection operator are explored in detail in the following section.

\section{Properties of Convex Sets and Functions}
\begin{defn}
A set $E \subset\mathbb{R}^k$ is {\em convex} if
\[\lambda x +(1-\lambda)y \in E\]
whenever $x \in E$, $y \in E$, and $0\leq \lambda \leq 1$
\end{defn}
\begin{rem}
Essentially, a convex set has the following property. For any two points $x, y \in E$ where $E$ is convex, all the points on the connecting line from $x$ to $y$ are also in $E$.
\end{rem}

\begin{defn}
A function $f: \mathbb{R}^k \rightarrow \mathbb R$ is {\em convex} if \[ f(\lambda x +(1-\lambda)y) \leq \lambda f(x) + (1-\lambda)f(y)\] $\forall 0\leq\lambda \leq1$.
\end{defn}

\begin{lem}
Let $f(\theta):\mathbb R^k\rightarrow \mathbb R$ be a convex function. Then for any constant $\delta >0$ the subset $\Omega_\delta=\lbrace{\theta \in \mathbb R^k | f(\theta)\leq\delta\rbrace}$ is convex.
\end{lem}
\begin{proof}
Let $\theta_1, \theta_2 \in \Omega_\delta$. Then $f(\theta_1)\leq\delta$ and $f(\theta_2)\leq\delta$. Since $f(x)$ is convex then for any $0\leq\lambda\leq 1$
\[f\bigl(\underbrace{\lambda \theta_1+(1-\lambda)\theta_2}_\theta\bigr)\leq\lambda \underbrace{f(\theta_1)}_{\leq\delta}+(1-\lambda)\underbrace{f(\theta_2)}_{\leq\delta}\leq\lambda\delta+(1-\lambda)\delta=\delta\]
$\therefore f(\theta)\leq\delta$, thus, $\theta\in \Omega_\delta$.
\end{proof}

\begin{lem}\label{lem:proj_diff}
Let $f(\theta):\mathbb R^k \rightarrow \mathbb R$ be a continuously differentiable convex function. Choose a constant $\delta>0$ and consider $\Omega_\delta=\lbrace{\theta \in \mathbb R^k | f(\theta)\leq\delta\rbrace}\subset \mathbb R$. Let $\theta^*$ be an interior point of $\Omega_\delta$, i.e. $f(\theta^*)<\delta$. Choose $\theta_b$ as a boundary point so that $f(\theta_b)=\delta$. Then the following holds:
\begin{equation}\label{eq:grad_geo}(\theta^*-\theta_b)^T \nabla f(\theta_b)\leq0\end{equation}
where $\nabla f(\theta_b)=\left(\frac{\partial f(\theta)}{\partial \theta_1 }\; \cdots \; \frac{\partial f(\theta)}{\partial \theta_k } \right)^T $ evaluated at $\theta_b$.
\end{lem}

\begin{proof}
$f(\theta)$ is convex $\therefore$
\[
f\left(\lambda\theta^*+(1-\lambda)\theta_b\right)\leq\lambda f(\theta^*)+(1-\lambda)f(\theta_b)
\]
equivalently,
\[
f\left(\theta_b +\lambda(\theta^*-\theta_b)\right)\leq f (\theta_b)+\lambda \left(f(\theta^*)-f(\theta_b)\right)
\]
For any $0<\lambda\leq1$:
\[
\frac{f\left(\theta_b +\lambda(\theta^*-\theta_b)\right)-f(\theta_b)}{\lambda}\leq f(\theta^*)-f(\theta_b)\leq \delta-\delta=0
\]
and taking the limit as $\lambda \rightarrow 0$ yields (\ref{eq:grad_geo}).
\end{proof}


\section{Projection}

\begin{defn}The {\em Projection Operator} for two vectors $\theta,y \in \mathbb R^k$ is now introduced as
\begin{equation}\label{eq:proj_vec2}
\text{Proj}(\theta,y,f)=\begin{cases} y-\frac{\nabla f(\theta) (\nabla f(\theta))^T}{\norm{\nabla f(\theta)}^2}yf(\theta)& \text{ if } f(\theta)>0 \wedge y^T\nabla f(\theta)>0\\
y & \text{ otherwise}.\end{cases}
\end{equation}
where $f:\mathbb R^k \rightarrow \mathbb R$ is a convex function and $\nabla f(\theta)=\left(\frac{\partial f(\theta)}{\partial \theta_1 }\; \cdots \; \frac{\partial f(\theta)}{\partial \theta_k } \right)^T$. Note that the following are notationally equivalent $\proj{\theta}{y}=\Proj{\theta}{y}{f}$ when the exact structure of the convex function $f$ is of no importance.
\end{defn}
\begin{rem} A geometrical interpretation of (\ref{eq:proj_vec2}) follows. Define a convex set $\Omega_0$ as
\begin{equation}\label{eq:omega1}\Omega_0 \triangleq \bigl\{ \theta \in \mathbb R^k | f(\theta)\leq 0\bigr\}
\end{equation}
and let $\Omega_1$ represent another convex set such that
\begin{equation}\label{eq:omega2}\Omega_1 \triangleq \bigl\{ \theta \in \mathbb R^k | f(\theta)\leq 1\bigr\}\end{equation}
From (\ref{eq:omega1}) and (\ref{eq:omega2}) $\Omega_0 \subset \Omega_1$. From the definition of the projection operator in (\ref{eq:proj_vecj}) $\theta$ is not modified when $\theta\in\Omega_0$. Let
\[
\Omega_{\mathcal A} \triangleq \Omega_1 \backslash \Omega_0=\bigl\{\theta|0<f(\theta)\leq1\bigr\}
\]
represent an annulus region. Within $\Omega_\mathcal{A}$ the
projection algorithm subtracts a scaled component of $y$ that is normal to boundary $\bigl\{ \theta |f(\theta)=\lambda\}$. When $\lambda=0$, the scaled normal component is 0, and when $\lambda=1$, the component of $y$ that is normal to the boundary $\Omega_1$ is entirely subtracted from $y$, so that $\Proj{\theta}{y}{f}$ is tangent to the boundary $\bigl\{ \theta |f(\theta)=1\bigr\}$. This discussion is visualized in Figure \ref{fig:proj_pic}.
\begin{figure}[!h]
\psfrag{y1}[cc][cc]{$y$}
\psfrag{y5}[cl][cl]{$\nabla f(\theta)$}
\psfrag{y2}[cr][cr]{Proj$(\theta,y)$}
\psfrag{y3}[cc][cc]{$\theta$}
\psfrag{y4}{$\theta^*$} 
\psfrag{o1}{$\Omega_{\mathcal A}$}
\psfrag{o0000000000004}[cr][cr]{$\{\theta | f(\theta)=0\}$}
\psfrag{o0000000000003}[cr][cr]{$\{\theta| f(\theta)=1\}$}
\centering
    \includegraphics[width=3in]{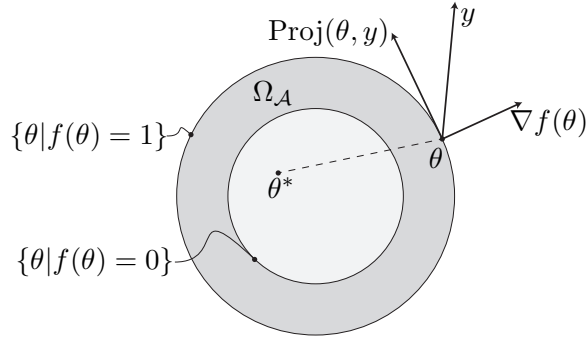}
\caption{Visualization of Projection Operator in $\mathbb R^2$.\label{fig:proj_pic}}
\end{figure}
\end{rem}

\begin{rem}Note that $(\nabla f(\theta))^T \text{Proj}(\theta,y)=0 \forall \theta$ when $f(\theta)=1$ and that the general structure of the algorithm is as follows
\begin{equation}
\text{Proj}(\theta,y)=y-\alpha(t)\nabla f(\theta)
\end{equation}
for some time varying $\alpha$ when the modification is triggered. Multiplying the left hand side of the equation by $(\nabla f(\theta))^T$ and solving for $\alpha$ one finds that 
\begin{equation}
\alpha(t)=\left((\nabla f(\theta))^T\nabla f(\theta) \right)^{-1}( \nabla f(\theta))^T y
\end{equation} and thus the algorithm takes the form 
\begin{equation}\label{rr}
\text{Proj}(\theta,y) =y - \nabla f(\theta) \left((\nabla f(\theta))^T\nabla f(\theta)\right)^{-1} (\nabla f(\theta))^T y f(\theta)
\end{equation}
where the modification is active. Notice that the $f(\theta)$ has been added to the definition, making \eqref{rr} continuous.
\end{rem}

\begin{lem}
One important property of the projection operator follows. Given $\theta^*\in \Omega_0$,
\begin{equation}\label{eq:prop1v}
(\theta-\theta^*)^T(\Proj{\theta}{y}{f}-y)\leq0.
\end{equation}
\end{lem}
\begin{proof}
Note that
\begin{equation*}
(\theta-\theta^*)^T(\Proj{\theta}{y}{f}-y) = (\theta^*-\theta)^T(y-\Proj{\theta}{y}{f})
\end{equation*}
If $ f(\theta)>0 \wedge y^T\nabla f(\theta)>0$, then
\begin{equation*}
 (\theta^*-\theta)^T\left(y-\left(y-\frac{\nabla f(\theta) (\nabla f(\theta))^T}{\norm{\nabla f(\theta)}^2}yf(\theta)\right)\right)
\end{equation*}
and using Lemma \ref{lem:proj_diff}
\begin{equation*}
 \frac{ \underbrace{(\theta^*-\theta)^T\nabla f(\theta)}_{\leq 0} \underbrace {(\nabla f(\theta))^Ty}_{> 0}}{\norm{\nabla f(\theta)}^2}\underbrace{f(\theta)}_{\geq 0} \leq 0
\end{equation*}
otherwise $\Proj{\theta}{y}{f}=y$.
\end{proof}

\begin{defn} [Projection Operator]
The general form of the projection operator is the $n\times m$ matrix extension to the vector definition above.
\[\text{Proj}(\Theta,Y,F) =\left[\Proj{\theta_1}{y_1}{f_1} \; \dotsc \; \Proj{\theta_m}{y_m}{y_m} \right]\]
where $\Theta=[\theta_1\; \dotsc \; \theta_m]\in \mathbb R^{n\times m}, Y=[y_1 \; \dotsc \; y_m]\in \mathbb R^{n\times m}$, and $F=[f_1(\theta_1)\; \dotsc \;f_m (\theta_m)]^T\in \mathbb R^{m\times 1}$. Recalling (\ref{eq:proj_vec2})
\begin{equation*}\label{eq:proj_vecj}
\text{Proj}(\theta_j,y_j,f_j)=\begin{cases} y_j-\frac{\nabla f_j(\theta_j) (\nabla f_j(\theta_j))^T}{\norm{\nabla f_j(\theta_j)}^2}y_jf_j(\theta_j)& \text{ if } f_j(\theta_j)>0 \wedge y_j^T\nabla f_j(\theta_j)>0\\
y_j & \text{ otherwise}\end{cases}
\end{equation*}
$j=1 \text{ to } m$.\end{defn}

\begin{lem}
Let $F=[f_1\; \dotsc \;f_m]^T\in \mathbb R^{m\times 1}$ be a convex vector function and $\hat\Theta=[\hat \theta_1 \; \dotsc \; \hat \theta_m ], \Theta=[\theta_1 \; \dotsc \;  \theta_m ], Y=[y_1 \; \dotsc \; y_m]$ where $\hat\Theta, \Theta, Y \in \mathbb R^{n \times m}$ then,
\begin{equation*}
\text{trace}\left \{ \bigl(\hat\Theta-\Theta\bigr)^T\bigl(\Proj{\hat \Theta}{Y}{F}-Y\bigr)\right\}\leq0.
\end{equation*}
\end{lem}
\begin{proof}
Using (\ref{eq:prop1v}),
\begin{equation*}\label{eq:prop1m}
\begin{split}
\text{trace}\left \{ \bigl(\hat\Theta-\Theta\bigr)^T\bigl(\Proj{\hat \Theta}{Y}{F}-Y\bigr)\right\}=&\sum_{j=1}^m(\hat\theta_j-\theta_j)^T (\Proj{\hat\theta_j}{y_j}{f_j}-y_j) \\ & \leq 0.
\qedhere\end{split}\end{equation*}
\end{proof}

The application of the projection algorithm in adaptive control is explored below.
\begin{lem}\label{lem:init}If an initial value
  problem, i.e. adaptive control algorithm with adaptive law and
  initial conditions, is defined by:
\begin{enumerate}
    \item $\dot \theta= Proj(\theta,y,f)$
    \item $\theta(t=0)=\theta_0\in\Omega_1=\{\theta\in\mathbb R^k | f(\theta)\leq 1\}$
    \item $f(\theta):\mathbb R^k \rightarrow \mathbb R$ is convex
\end{enumerate}
     Then $\theta(t)\in\Omega_1  \forall t\geq 0$.
\end{lem}
\begin{proof}
Taking the time derivative of the convex function
\begin{equation}\label{eq:dot_f}
\dot f(\theta)=(\nabla f(\theta))^T\dot \theta=(\nabla f(\theta))^T\Proj{\theta}{y}{f}
\end{equation}
Substitution of (\ref{eq:dot_f}) into (\ref{eq:proj_vec2}) leads to
\begin{equation*}\begin{split}
\dot f(\theta)&=(\nabla f(\theta))^T\Proj{\theta}{y}{f}\\
&=\begin{cases} (\nabla f(\theta))^Ty(1-f(\theta))& \text{ if } f(\theta)>0 \wedge y^T\nabla f(\theta)>0\\
(\nabla f(\theta))^Ty & \text{ if } f(\theta)\leq0 \vee y^T\nabla f(\theta)\leq0    \end{cases}
\end{split}\end{equation*}
therefore
\begin{equation*}
\begin{cases} \dot f(\theta) >0 & \text{ if } 0< f(\theta)< 1 \wedge y^T\nabla f(\theta)>0\\
\dot f(\theta)=0 & \text{ if }f(\theta)=1 \wedge  y^T\nabla f(\theta)>0  \\
 \dot f(\theta)<0 & \text{ if }f(\theta)\leq 0 \vee  y^T\nabla f(\theta)\leq 0  \end{cases}.
\end{equation*}
Thus $f(\theta_0)\leq 1 \Rightarrow f(\theta)\leq 1 \forall t\geq0$.
\end{proof}
\begin{rem}
Given $\theta_0\in \Omega_0$, $\theta$ may increase up to the boundary where $f(\theta)=1$. However, $\theta$ never leaves the convex set $\Omega_1$.
\end{rem}
\begin{exmp}[Projection Algorithm in Adaptive Control Law]\label{ex:adaptive}
Let $\Theta(t):\mathbb R^+\rightarrow \mathbb R^{m\times n}$ represent a time varying feedback gain in a dynamical system. This feedback gain is implemented as:
\begin{equation*}
u=\Theta(t)^Tx
\end{equation*}
where $u\in\re^n$ represents the control input and $x\in\re^m$ the state vector. The time varying feedback gain is adjusted using the following adaptive law
\begin{equation*}\label{eq:adaptive_ex_law}
\dot\Theta=\Proj{\Theta}{-x e^T PB}{F}
\end{equation*}
where $e\in{\re^m}$ is an error signal in the state vector space, $P\in\re^{m\times m}$ is a square matrix derived from a Lyapunov relationship and $B\in{\re^{m\times n}}$ is the input Jacobian for the LTI system to be controlled and $F(\Theta)=[f_1(\theta_1) \; \dotsc  \; f_m(\theta_m))]^T$. The projection algorithm operates with the family of convex functions
\begin{equation*}\label{eq:example_f}
f(\theta ; \vartheta,\varepsilon)=\frac {\norm{\theta}^2-{\vartheta}^2}
                 {2\varepsilon\vartheta+\varepsilon^2}.
\end{equation*}
Then, the components of the convex vector function $F$ are chosen as
\begin{equation}\label{eq:example_fi}
f_i(\theta_i)=f(\theta_i;\vartheta_i,\varepsilon_i).
\end{equation}
Each $i$--th component of $F$ is associated with two constant scalar quantities $\vartheta_i$ and $\varepsilon_i$. From (\ref{eq:example_fi}), $f_i(\theta_i)=0$ when $\norm{\theta_i}=\vartheta_i$, and $f_i(\theta_i)=1$ when $\norm{\theta_i}=\vartheta_i+\varepsilon_i$. If the initial condition for $\Theta$ is such that $\Theta(t=0)\in \Theta_0=[\theta_{0,1} \;\dotsc\; \theta_{0,m} ]$ where $\{\theta_{0,i}|f_i(\theta_i)\leq0\,i=1\text{ to } m\}$, then each $\theta_i$ satisfies all three conditions for Lemma \ref{lem:init}. Thus $\norm {\theta_i(t)} \leq \vartheta_i +\epsilon_i \forall t\geq 0$.
\end{exmp}

\section{$\Gamma$--Projection}
\begin{defn}
A variant of the projection algorithm, {$\Gamma$--projection}, updates the parameter along a symmetric positive definite gain $\Gamma$ as defined below
\begin{equation}\label{eq:proj_vec_gamma}
\text{Proj}_\Gamma(\theta,y,f)=\begin{cases} \Gamma y- \Gamma\frac{\nabla f(\theta) (\nabla f(\theta))^T}{(\nabla f(\theta))^T \Gamma \nabla f(\theta)} \Gamma yf(\theta)& \text{ if } f(\theta)>0 \wedge y^T\Gamma \nabla f(\theta)>0\\
\Gamma y & \text{ otherwise}.\end{cases}
\end{equation}
This method was first introduced in \cite{ioabook}.
\end{defn}

\begin{lem}
Given $\theta^*\in \Omega_0$,
\begin{equation}\label{eq:prop1vg}
(\theta-\theta^*)^T(\Gamma^{-1}\text{Proj}_\Gamma({\theta},{y},{f})-y)\leq 0.
\end{equation}
\end{lem}
\begin{proof}
If $ f(\theta)>0 \wedge y^T\Gamma \nabla f(\theta)>0$, then
\begin{equation*}
 (\theta^*-\theta)^T\left(y-\Gamma^{-1}\left(\Gamma y- \Gamma\frac{\nabla f(\theta) (\nabla f(\theta))^T}{(\nabla f(\theta))^T \Gamma \nabla f(\theta)} \Gamma yf(\theta)\right)\right)
\end{equation*}
and using Lemma \ref{lem:proj_diff}
\begin{equation*}
 \frac{ \underbrace{(\theta^*-\theta)^T\nabla f(\theta)}_{\leq 0} \underbrace {(\nabla f(\theta))^T\Gamma y}_{> 0}}{(\nabla f(\theta))^T \Gamma \nabla f(\theta)}\underbrace{f(\theta)}_{\geq 0} \leq 0
\end{equation*}
otherwise $\text{Proj}_\Gamma({\theta},{y},{f})=\Gamma y$.
\end{proof}


\bibliographystyle{plain}
\bibliography{biblio}

\end{document}